\documentclass[a4paper,fleqn]{cas-sc}

\usepackage[authoryear,longnamesfirst]{natbib}
\usepackage{amsthm}

\theoremstyle{plain}
\newtheorem{theorem}{Theorem}
\newtheorem{assumption}[theorem]{Assumption}

\def\tsc#1{\csdef{#1}{\textsc{\lowercase{#1}}\xspace}}
\tsc{WGM}
\tsc{QE}

\begin{document}
\let\WriteBookmarks\relax
\def\floatpagepagefraction{1}
\def\textpagefraction{.001}

\shorttitle{}    

\shortauthors{}  

\title[mode = title]{Continuity of VaR and Continuous Differentiability of CVaR under Decision-Dependent Losses
}

\author[1,2]{Amal Sakr}
\cormark[1]
\ead{amal.sakr@telecom-sudparis.eu}

\author[1]{Andrea Araldo}

\author[2]{Tamer Ba\c{s}ar}

\author[1]{Tijani Chahed}

\affiliation[1]{
    organization={Institut Polytechnique de Paris},
    city={Palaiseau},
    postcode={91120},
    country={France}
}

\affiliation[2]{
    organization={University of Illinois Urbana-Champaign},
    city={Urbana},
    postcode={61801},
    state={IL},
    country={USA}
}

\cortext[1]{Corresponding author}

\begin{abstract}
Value-at-risk (VaR) and conditional value-at-risk (CVaR) are widely used in
risk-aware optimization and equilibrium models. When the loss depends on a
decision variable, the induced distribution, the VaR threshold, and the CVaR
tail set all change with the decision. This makes the regularity of the VaR
and CVaR maps nontrivial. We give simple sufficient conditions under which the
VaR map is continuous and the corresponding CVaR map is continuously
differentiable. The assumptions are local around the VaR level and rely on
dominated pathwise differentiability of the scenario-wise loss. We also derive
the CVaR gradient formula, thereby justifying first-order analysis for
decision-dependent tail-risk models.
\end{abstract}




\maketitle


\section{Introduction}

Conditional value-at-risk (CVaR), also known as expected shortfall or
average value-at-risk, is a standard risk measure for decision problems in
which tail losses are more important than average performance alone
\citep{rockafellar2000optimization,acerbi2002expected,hong2009simulating}.
It has been widely used in portfolio optimization and financial risk
management, and it has also been incorporated into optimization models beyond
finance \citep{rockafellar2000optimization,filippi2020conditional}. CVaR has
also been used in risk-sensitive learning, where policy-gradient methods have
been developed for coherent risk measures including CVaR
\citep{tamar2015policy}. More recently, CVaR-based variational inequalities
have been studied, with stochastic-approximation and sample-average schemes
proposed for their solution
\citep{verbree2020stochastic,cherukuri2024sample}.

These applications often rely on sensitivity information for CVaR. \citet{hong2009simulating} study Monte Carlo estimation of CVaR sensitivities and show that such
sensitivities can be represented through conditional expectations. Gradient-based CVaR optimization also relies on
estimating CVaR sensitivities with respect to the design variable
\citep{ganesh2023gradient}. In risk-sensitive learning, policy-gradient methods
are developed to optimize coherent risk criteria, including CVaR-type criteria,
through first-order information \citep{tamar2015policy}. Similarly, in
CVaR-based variational inequalities, the map defining the equilibrium problem
is built from CVaR evaluations of uncertain functions, and
stochastic-approximation or sample-average schemes are used to solve the
resulting problems
\citep{verbree2020stochastic,verbree2022stochastic,cherukuri2024sample}.
These works show that CVaR regularity is important when CVaR appears inside an
optimization, learning, or equilibrium formulation.

However, such regularity cannot be taken for granted. In \citep{rockafellar2000optimization}, CVaR involves a positive-part function,
and CVaR minimization has been treated as a nonsmooth stochastic optimization
problem. Existing works therefore use smoothing, nonsmooth optimization, or
gradient-free methods when direct smooth first-order structure is not available
\citep{tarnopolskaya2010cvar,chaudhuri2020risk}. The issue becomes more
delicate when the loss itself depends on the decision variable. In that case,
changing the decision changes not only the scenario-wise loss, but also the
induced distribution, the VaR threshold, and the tail set used in the CVaR
evaluation. Thus, even when the scenario-wise loss is smooth, the resulting
CVaR map is not simply an expectation of a smooth random function. Its
differentiability requires separate justification.

The VaR threshold is a central part of this difficulty. VaR is known to be
analytically delicate: it may fail to be continuous and differentiable
\citep{balbas2017differential}, and it can be ill-behaved as a function of the
decision variables \citep{rockafellar2000optimization}. Moreover, CVaR
sensitivity formulas commonly require regularity of the VaR threshold; for
instance, \citet{hong2011monte} assume differentiability of the VaR level when
deriving CVaR sensitivity. Therefore, establishing continuity of the VaR map is
a natural preliminary step before proving differentiability of the
decision-dependent CVaR map.

This note provides such a justification under simple local assumptions. First,
we prove that the VaR map is continuous under local distributional regularity
around the VaR level. This controls the movement of the endogenous tail
threshold. Second, under an additional dominated pathwise differentiability
condition on the scenario-wise loss, we prove that the corresponding CVaR map
is continuously differentiable and derive its gradient formula. 

\section{Preliminaries and notation}
\label{sec:prelim}

Let \((\Omega,\mathcal F,\mathbb P)\) be a probability space and let
\(O\subseteq\mathbb R^d\) be an open set. The variable \(x\in O\) denotes a
decision parameter. For each decision \(x\), the model generates a real-valued
loss random variable \(L(x)\in L^1(\Omega)\). For a scenario
\(\omega\in\Omega\), \(L_\omega(x)\) denotes the realized loss under decision
\(x\). Thus \(x\mapsto L_\omega(x)\) describes how the loss in a fixed scenario
changes with the decision, while \(L(x)\) denotes the induced random loss.

For \((z,x)\in\mathbb R\times O\), define the cumulative
distribution function (CDF)  of the
decision-dependent loss by
\[
F(z,x):=\mathbb P(L(x)\le z).
\]
For a fixed confidence level \(\alpha\in(0,1)\), define
\[
\operatorname{VaR}_{\alpha}(L(x))
:=\inf\{z\in\mathbb R:\,F(z,x)\ge\alpha\}.
\]
Hence \(\operatorname{VaR}_{\alpha}(L(x))\) is the \(\alpha\)-quantile of the loss distribution induced by
the decision \(x\).

We use the Rockafellar--Uryasev representation of CVaR
\citep{rockafellar2000optimization,rockafellar2002conditional}. For
\((x,\eta)\in O\times\mathbb R\), set
\[
\phi(x,\eta):=
\eta+\frac{1}{1-\alpha}\mathbb E[(L(x)-\eta)_+],
\qquad (a)_+:=\max\{a,0\}.
\]
Then
\[
\operatorname{CVaR}_{\alpha}(L(x))
=
\inf_{\eta\in\mathbb R}\phi(x,\eta).
\]

\section{Regularity assumptions}
\label{sec:assumptions}

We impose two assumptions. The first controls the distribution of the
decision-dependent loss near the VaR level.

\begin{assumption}
\label{ass:var_continuity}
For every \(x\in O\), the random variable \(L(x)\) is integrable. Moreover: \begin{enumerate} \item the CDF  \( F(z,x)\) is jointly continuous in \((z,x)\in\mathbb R\times O\); \item for every \(x\in O\), the random variable \(L(x)\) admits a density \(f(\cdot;x)\) in a neighborhood of \(\operatorname{VaR}_{\alpha}(L(x)), \) and this density is strictly positive in that neighborhood. \end{enumerate} \end{assumption}

Assumption~\ref{ass:var_continuity} excludes jumps and flat portions of the loss
distribution around the VaR threshold. It is local: positivity of the density
is required only near the relevant quantile, not over the whole support. This
is the condition used to prove that the VaR threshold moves continuously with
the decision variable.

The second assumption controls the pathwise dependence of the loss on the
decision variable.

\begin{assumption}[Dominated pathwise differentiability]
\label{ass:cvar}
The loss satisfies:
\begin{enumerate}
\item \(x\mapsto L_\omega(x)\) is continuously differentiable on \(O\) for
a.e. \(\omega\in\Omega\);
\item there exists \(M\in L^1(\Omega)\) such that
\[
\sup_{x\in O}\|\nabla_x L_\omega(x)\|\le M_{\omega},
\qquad \text{for a.e. }\omega\in\Omega .
\]
\end{enumerate}
\end{assumption}

Assumption~\ref{ass:cvar} is a global dominated differentiability condition on \(O\).
It allows differentiation to pass through the expectation in the
Rockafellar--Uryasev representation.

Together, these assumptions separate the two sources of difficulty:
Assumption~\ref{ass:var_continuity} controls the movement of the
endogenous quantile \(v(x)\), while
Assumption~\ref{ass:cvar} controls the sensitivity of the
scenario-wise loss \(L_\omega(x)\). Under these conditions, the next section
shows that \(v\) is continuous and that
\(\mathrm{CVaR}_{\alpha}(L(x))\) is continuously differentiable.

\section{Main results}
\label{sec:main_results}

We now present the main results of the paper. 

\begin{theorem}
\label{thm:var_continuity}
Suppose that Assumption~\ref{ass:var_continuity} holds. Then the
mapping
\[
x\mapsto \mathrm{VaR}_{\alpha}(L(x))
\]
is continuous on \(O\).
\end{theorem}

Theorem~\ref{thm:var_continuity} shows that the risk threshold is stable
under perturbations of the decision parameter. This property is essential:
without continuity of \(v(x)\), small changes in \(x\) could move the
value-at-risk threshold discontinuously and therefore change the CVaR tail
region abruptly.
\begin{proof}
    Fix \(x\in \mathcal O\), and set
\begin{equation}\label{eq:bar_v_def}
\bar v
:=
\mathrm{VaR}_{\alpha}\!\left(L(x)\right)
\end{equation}
By definition,
\begin{equation}\label{eq:bar_v_min_def}
\bar v
=
\min\left\{
z\in\mathbb R:
F(z,x)\ge \alpha
\right\}
\end{equation}
Hence
\(
F(\bar v,x)\ge \alpha
\).
It is easy to show that\footnote{Suppose, by contradiction, that
\(F(\bar v,x)>\alpha\). Since the CDF is continuous in
\(z\), there exists \(\varepsilon>0\) such that
\(F(\bar v-\varepsilon,x)> \alpha\). But
\(\bar v-\varepsilon<\bar v\), which contradicts the minimality of \(\bar v\)
in \eqref{eq:bar_v_min_def}. Hence
\(F(\bar v,x)=\alpha\).}
\begin{equation}\label{eq:F_bar_v_eq_alpha}
F(\bar v,x)=\alpha
\end{equation}

Moreover, by Assumption~\ref{ass:var_continuity}(2), the CDF
\(z\mapsto F(z,x)\) is strictly increasing, and therefore
invertible with respect to \(z\). Since \(\bar v\) satisfies
\(F(\bar v,x)=\alpha\), the equation
\[
F(z,x)=\alpha
\]
has the unique solution \(z=\bar v\).

Now let \(\varepsilon>0\). Since, for fixed
\(x\), the CDF
\(z\mapsto F(z,x)\) is strictly increasing with respect to
\(z\), and since \eqref{eq:F_bar_v_eq_alpha} holds, we have
\begin{equation}\label{eq:strict_bounds_at_bar_h}
F(\bar v-\varepsilon,x)
<
\alpha
<
F(\bar v+\varepsilon,x).
\end{equation}

By the joint continuity of \(F\), there exists \(\delta>0\) such that whenever
\(\|y-x\|<\delta\), one has
\begin{equation}\label{eq:strict_bounds_near_bar_h}
F(\bar v-\varepsilon,y)
<
\alpha
<
F(\bar v+\varepsilon,y).
\end{equation}

Let
\[
v(y)
:=
\mathrm{VaR}_{\alpha}\!\left(L(y)\right).
\]
By the same argument as in \eqref{eq:F_bar_v_eq_alpha}, we have
\begin{equation}\label{eq:F_v_h_eq_alpha}
F\!\left(v(y),y\right)=\alpha.
\end{equation}

Since \(z\mapsto F(z,y)\) is strictly increasing, the inequalities in
\eqref{eq:strict_bounds_near_bar_h} imply
\[
\bar v-\varepsilon
<
v(y)
<
\bar v+\varepsilon.
\]


Thus
\[
\left|
v(y)-v(x)
\right|
=
\left|
v(y)-\bar v
\right|
< \varepsilon .
\]
Since \(\varepsilon>0\) was arbitrary, \(v\) is continuous at
\(x\). Since \(x\in \mathcal O\) was arbitrary, the
mapping
\[
x\mapsto
\mathrm{VaR}_{\alpha}\!\left(L(x)\right)
\]
is continuous on \(\mathcal O\).
\end{proof}
\begin{theorem}
\label{thm:cvar_differentiability}
Suppose that Assumptions~\ref{ass:var_continuity}
and~\ref{ass:cvar} hold. Then the mapping
\(
x\mapsto \mathrm{CVaR}_{\alpha}(L(x))
\)
is continuously differentiable on \(O\). Moreover,
\[
\nabla_x\operatorname{CVaR}_{\alpha}(L(x))
=
\frac{1}{1-\alpha}
\mathbb E\!\left[
\nabla_x L_\omega(x)\mathbf 1_{\{L_\omega(x)>v(x)\}}
\right].
\]
\end{theorem}

Theorem~\ref{thm:cvar_differentiability} is the central result. It shows
that, under mild local regularity conditions, the CVaR of a
decision-dependent loss is a continuously differentiable function of the
decision variable. 

\begin{proof}
For
a.e. \(\omega\in\Omega\), the map
\(
x\mapsto L_{\omega}  (x)
\)
is continuously differentiable on \(\mathcal O\). 
For \((x,\eta)\in O\times\mathbb R\), define
\[
\phi (x,\eta)
:=
\eta+\frac{1}{1-\alpha}\,
\mathbb E_\omega\!\left[
\bigl(L_{\omega}  (x)-\eta\bigr)_+
\right]
\]
where \((x)_+ := \max\{x,0\}\) denotes the positive part of \(x\).
By Assumption~\ref{ass:cvar}(1), the expectation above is finite for every
\((x,\eta)\in O\times\mathbb R\). Therefore, by~\cite{rockafellar2000optimization} and~\cite[Eq. (32)]{li2022risk},
\begin{equation}\label{eq:ru_rep_revised}
\mathrm{CVaR}_{\alpha}\!\bigl(L (x)\bigr)
=
\min_{\eta\in\mathbb R}\phi (x,\eta),
\qquad \forall x\in \mathcal O
\end{equation}
By the Rockafellar--Uryasev representation
\cite[Theorem~1 and its proof]{rockafellar2000optimization},
\(v(x)\) is a minimizer of
\(\eta\mapsto \phi (x,\eta)\).

We proceed as follows.

\paragraph{\textbf{1. Differentiability of \texorpdfstring{\(\phi \)}{phi  } with respect to \texorpdfstring{\(x\)}{h}}.}
\label{sec:diff_phi}

We now show that, for any fixed $\eta\in\mathbb R$, function $x\to \phi (x, \eta)$ is differentiable.

Fix \(\eta\in\mathbb R\). For every fixed \(x\in \mathcal O\), consider the map
\begin{equation}
    x\mapsto
\bigl(L_{\omega}  (x)-\eta\bigr)_+ 
\label{eq:map}
\end{equation}
There are three cases.

If
\[
L_{\omega}  (x)-\eta<0,
\]
then, by continuity of the function
\(x\mapsto L_{\omega}  (x)\), there exists a
neighborhood of \(x\) on which
\(
L_{\omega}  (x)-\eta<0.
\)
Hence
\(
\bigl(L_{\omega}  (x)-\eta\bigr)_+=0
\)
in this neighborhood. Therefore, the map~\eqref{eq:map} is locally constant and differentiable
with derivative zero.

If
\[
L_{\omega}  (x)-\eta>0,
\]
then, by continuity of the function
\(x\mapsto L_{\omega}  (x)\), there exists a
neighborhood of \(x\) on which
\(
L_{\omega}  (x)-\eta>0.
\)
Hence
\(
\bigl(L_{\omega}  (x)-\eta\bigr)_+
=
L_{\omega}  (x)-\eta
\)
in this neighborhood. Since the function
\(x\mapsto L_{\omega}  (x)\) is differentiable, the map~\eqref{eq:map} is
differentiable at \(x\).

The only remaining case is
\[
L_{\omega}  (x)-\eta=0
\]
At such points, the map~\eqref{eq:map}
may fail to be differentiable at the fixed point \(x\). However, by
Assumption~\ref{ass:var_continuity}(1), the CDF of
\(L (x)\) is continuous. Hence
\[
\mathbb P\bigl(L (x)=\eta\bigr)=0
\]
Therefore, this case occurs with probability zero and can be neglected.

Therefore, for every fixed \(x\in O\), the map
\[
x\mapsto
\bigl(L_{\omega}  (x)-\eta\bigr)_+
\]
is differentiable at \(x\) for a.e.
\(\omega\in\Omega\), with
\begin{equation}
\nabla_{x}
\bigl(L_{\omega}  (x)-\eta\bigr)_+
=
\nabla_{x}L_{\omega}  (x)\,
\mathbf 1_{\{L_{\omega}  (x)>\eta\}} .
\label{eq:nabla_l}
\end{equation}
Indeed, using \eqref{eq:nabla_l} and the fact that the indicator function
takes values only in \(\{0,1\}\), we have
\[
\begin{aligned}
\left\|
\nabla_{x}\bigl(L_{\omega}  (x)-\eta\bigr)_+
\right\|
&=
\left\|
\nabla_{x}L_{\omega}  (x)\,
\mathbf 1_{\{L_{\omega}  (x)>\eta\}}
\right\| \\
&\le
\left\|
\nabla_{x}L_{\omega}  (x)
\right\| \\
&\le
\sup_{x\in O}
\left\|
\nabla_{x}L_{\omega}  (x)
\right\| \\
&\le M_{\omega}  ,
\qquad \text{for a.e. \(\omega\in\Omega\)}
\end{aligned}
\]
where the last inequality follows
from Assumption~\ref{ass:cvar}(2).

Since \(M \) is an integrable random function, differentiation may be passed under the
expectation, yielding
\begin{equation}\label{eq:grad_phi_revised}
\nabla_{x}\phi (x,\eta)
=
\frac{1}{1-\alpha}\,
\mathbb E_\omega\!\left[
\nabla_{x}L_{\omega}  (x)\,
\mathbf 1_{\{L_{\omega}  (x)>\eta\}}
\right]
\end{equation}
see~\cite[Theorem (Differentiation under the integral sign)]{chua2016probabilitymeasure}.

\vspace{0.5cm}

\paragraph{\textbf{2. Continuity of the gradient}.}
\label{sec:continuity-of-the-gradient}

We now show that
\(
(x,\eta)\mapsto \nabla_{x}\phi (x,\eta)
\)
is continuous on \(\mathcal O\times\mathbb R\). 

Let
\((x_n,\eta_n)\to(x,\eta)\). For a.e.
\(\omega\in\Omega\), continuity of
\(x\mapsto \nabla_{x}L_{\omega}  (x)\) gives
\(
\nabla_{x}L_{\omega}  (x_n)
\to
\nabla_{x}L_{\omega}  (x),
\)
and continuity of \(x\mapsto L_{\omega}  (x)\) gives
\(
L_{\omega}  (x_n)\to L_{\omega}  (x).
\)
Hence
\(
\mathbf 1_{\{L_{\omega}  (x_n)>\eta_n\}}
\to
\mathbf 1_{\{L_{\omega}  (x)>\eta\}}
\)
for a.e. \(\omega\), except possibly on
\(\{L_{\omega}  (x)=\eta\}\), which has probability zero. Since
\[
\left\|
\nabla_{x}L_{\omega}  (x_n)\,
\mathbf 1_{\{L_{\omega}  (x_n)>\eta_n\}}
\right\|
\le M_{\omega}  ,
\]
Then \cite[Theorem 16.4]{billingsley2012probability} implies
\[
\nabla_{x}\phi (x_n,\eta_n)
\to
\nabla_{x}\phi (x,\eta).
\]
Thus \((x,\eta)\mapsto \nabla_{x}\phi (x,\eta)\)
is continuous on \(O\times\mathbb R\).

\paragraph{\textbf{3. Continuity of the Value at Risk}}\label{sec:cont_var_in_proof}

Define
\[
v(x):=
\mathrm{VaR}_{\alpha}\!\bigl(L (x)\bigr),
\qquad x\in O
\]
By Theorem~\ref{thm:var_continuity}, the mapping
\(
x \mapsto 
v(x)
\)
is continuous on \(O\).

\paragraph{\textbf{4. Differentiability of CVaR}}\label{sec:diff_cvar}

Define
\begin{equation}
g (x)
:=
\mathrm{CVaR}_{\alpha}\!\bigl(L (x)\bigr)
=
\phi \bigl(x,v(x)\bigr)
\label{eq:def_g_cvar}
\end{equation}
We show that \(g \) is differentiable on \(O\).
Fix \(x\in \mathcal O\), and let \(\mathbf u\in\mathbb R^d\).
To prove differentiability of \(g \) at \(x\), we
show that there exists a linear map \(J:\mathbb R^d\to\mathbb R\) such that
\begin{equation}
\lim_{\mathbf u\to 0}
\frac{
\left|
g (x+\mathbf u)
-
g (x)
-
J(\mathbf u)
\right|
}{
\|\mathbf u\|
}
=0
\label{eq:differentiability_def}
\end{equation}
We will show that this holds with
\begin{equation}
J(\mathbf u)
=
\left\langle
\nabla_{x}\phi 
\bigl(x,\eta\bigr),
\mathbf u
\right\rangle,
\quad \text{with } \eta = v(x)
\label{eq:def_J}
\end{equation}
Equivalent to~\eqref{eq:differentiability_def}, we will prove that
\begin{equation}
g (x+\mathbf u)-g (x)
=
\left\langle
\nabla_{x}\phi 
\bigl(x,v(x)\bigr),
\mathbf u
\right\rangle
+
o(\|\mathbf u\|)
\label{eq:target_expansion}
\end{equation}
Since
\begin{equation}
v(x+\mathbf u)
\in
\arg\min_{\eta\in\mathbb R}
\phi (x+\mathbf u,\eta)
\label{eq:v_min_hu}
\end{equation}
we have
\begin{equation}
\phi \bigl(x+\mathbf u,v(x+\mathbf u)\bigr)
\le
\phi \bigl(x+\mathbf u,v(x)\bigr)
\label{eq:min_ineq_hu}
\end{equation}
and since
\begin{equation}
v(x)
\in
\arg\min_{\eta\in\mathbb R}
\phi (x,\eta),
\label{eq:v_min_h}
\end{equation}
we have
\begin{equation}
\phi \bigl(x,v(x)\bigr)
\le
\phi \bigl(x,v(x+\mathbf u)\bigr)
\label{eq:min_ineq_h}
\end{equation}
Therefore, by \eqref{eq:min_ineq_hu} and \eqref{eq:min_ineq_h},
\begin{align}
&\phi \bigl(x+\mathbf u,v(x+\mathbf u)\bigr)
-
\phi \bigl(x,v(x+\mathbf u)\bigr) \quad\le
\underbrace{
\phi \bigl(x+\mathbf u,v(x+\mathbf u)\bigr)
}_{=\,g (x+\mathbf u)}
-
\underbrace{
\phi \bigl(x,v(x)\bigr)
}_{=\,g (x)} \le
\phi \bigl(x+\mathbf u,v(x)\bigr)
-
\phi \bigl(x,v(x)\bigr)
\label{eq:sandwich_before_subtract}
\end{align}
Subtracting from the three terms in~\eqref{eq:sandwich_before_subtract}
\(
\left\langle
\nabla_{x}\phi \bigl(x,v(x)\bigr),
\mathbf u
\right\rangle
\)
\begin{equation}
A(\mathbf u)
\le
g (x+\mathbf u)-g (x)
-
\left\langle
\nabla_{x}\phi \bigl(x,v(x)\bigr),
\mathbf u
\right\rangle
\le
B(\mathbf u)
\label{eq:sandwich_AB}
\end{equation}
where
\begin{align}
A(\mathbf u)
:=
\phi \bigl(x+\mathbf u,v(x+\mathbf u)\bigr)
-
\phi \bigl(x,v(x+\mathbf u)\bigr)-
\left\langle
\nabla_{x}\phi \bigl(x,v(x)\bigr),
\mathbf u
\right\rangle,
\label{eq:def_A}
\\[1ex]
B(\mathbf u)
:=
\phi \bigl(x+\mathbf u,v(x)\bigr)
-
\phi \bigl(x,v(x)\bigr) \quad-
\left\langle
\nabla_{x}\phi \bigl(x,v(x)\bigr),
\mathbf u
\right\rangle
\label{eq:def_B}
\end{align}

We claim that
\begin{equation}
A(\mathbf u)=o(\|\mathbf u\|),
\qquad
B(\mathbf u)=o(\|\mathbf u\|),
\qquad
\text{as }\mathbf u\to 0.
\label{eq:AB_claim}
\end{equation}
Indeed, by Step~3,
\begin{equation}
v(x+\mathbf u)\to v(x)
\qquad\text{as }\mathbf u\to 0.
\label{eq:v_continuity}
\end{equation}
Moreover, the map
\begin{equation}
(x,\eta)\mapsto \nabla_{x}\phi (x,\eta)
\label{eq:grad_map}
\end{equation}
is continuous on \(\mathcal O\times\mathbb R\). 

Fix $\eta\in\mathbb R$. Since $\phi (\cdot,\eta)$ is continuously differentiable on $O$, and since the segment 
$\{x+s\mathbf u:\; s\in[0,1]\}$ is contained in $O$ for $\mathbf u$ sufficiently small, we may use~\cite[Theorem~6.7 and its proof]{strang2017calculus} to the function $x\mapsto \phi (x,\eta)$ along the path $\mathbf r(s)=x+s\mathbf u$. This yields
\begin{align}
\phi (x+\mathbf u,\eta)
-
\phi (x,\eta)
&=
\int_0^1
\left\langle
\nabla_{x}\phi (x+s\mathbf u,\eta),
\mathbf u
\right\rangle ds
\label{eq:fundamental_path}
\end{align}
\[
\iff
\]
\begin{align}
&\phi \bigl(x+\mathbf u,\eta\bigr)
-
\phi \bigl(x,\eta\bigr)
-
\left\langle
\nabla_{x}\phi \bigl(x,v(x)\bigr),
\mathbf u
\right\rangle =
\int_0^1
\left\langle
\nabla_{x}\phi (x+s\mathbf u,\eta)
-
\nabla_{x}\phi \bigl(x,v(x)\bigr),
\mathbf u
\right\rangle
\,ds
\label{eq:fundamental_path_subtracted}
\end{align}
In particular, for $\eta=v(x+\mathbf u)$, ~\eqref{eq:def_A} becomes
\begin{align}
A(\mathbf u)
&=
\int_0^1
\Big\langle
\begin{aligned}
&\nabla_{x}\phi 
\bigl(
x+s\mathbf u,
v(x+\mathbf u)
\bigr)\\
&-
\nabla_{x}\phi 
\bigl(
x,
v(x)
\bigr),
\mathbf u
\end{aligned}
\Big\rangle ds 
\label{eq:A_integral}
\end{align}
Similarly for $\eta=v(x)$, \eqref{eq:def_B} becomes
\begin{align}
B(\mathbf u)
&=
\int_0^1
\Big\langle
\begin{aligned}
&\nabla_{x}\phi 
\bigl(
x+s\mathbf u,
v(x)
\bigr)\\
&-
\nabla_{x}\phi 
\bigl(
x,
v(x)
\bigr),
\mathbf u
\end{aligned}
\Big\rangle ds 
\label{eq:B_integral}
\end{align}
We first prove that $A(\mathbf u)=o(\|\mathbf u\|)$. From
\eqref{eq:A_integral}, using the triangle inequality for integrals and~\cite[6.15]{axler2024linear}, we get
\begin{align}
|A(\mathbf u)|
&\le
\int_0^1
\Big|
\Big\langle
\begin{aligned}
&\nabla_{x}\phi 
\bigl(
x+s\mathbf u,
v(x+\mathbf u)
\bigr)\\
&-
\nabla_{x}\phi 
\bigl(
x,
v(x)
\bigr),
\mathbf u
\end{aligned}
\Big\rangle
\Big|\,ds
\le
\|\mathbf u\|
\int_0^1
\Big\|
\begin{aligned}
&\nabla_{x}\phi 
\bigl(
x+s\mathbf u,
v(x+\mathbf u)
\bigr)\\
&-
\nabla_{x}\phi 
\bigl(
x,
v(x)
\bigr)
\end{aligned}
\Big\|\,ds
\label{eq:A_bound}
\end{align}
Dividing \eqref{eq:A_bound} by $\|\mathbf u\|$ gives
\begin{align}
\frac{|A(\mathbf u)|}{\|\mathbf u\|}
&\le
\int_0^1
\Big\|
\begin{aligned}
&\nabla_{x}\phi 
\bigl(
x+s\mathbf u,
v(x+\mathbf u)
\bigr)\\
&-
\nabla_{x}\phi 
\bigl(
x,
v(x)
\bigr)
\end{aligned}
\Big\|\,ds
\label{eq:A_ratio}
\end{align}
Since \(\mathbf u\to 0\), we have
\begin{equation}
\sup_{s\in[0,1]}
\bigl\|
(x+s\mathbf u)-x
\bigr\|
=
\sup_{s\in[0,1]} s\,\|\mathbf u\|
\le
\|\mathbf u\|
\longrightarrow 0.
\label{eq:uniform_h_convergence}
\end{equation}
Hence,
\begin{equation}
x+s\mathbf u \to x
\quad \text{uniformly in } s\in[0,1].
\label{eq:h_uniform_statement}
\end{equation}

Moreover, since \(v\) is continuous at
\(x\),
\begin{equation}
v(x+\mathbf u)\to v(x).
\label{eq:v_uniform_statement}
\end{equation}
Therefore,
\begin{equation}
\bigl(x+s\mathbf u,\,
v(x+\mathbf u)\bigr)
\to
\bigl(x,\,v(x)\bigr)
\label{eq:pair_uniform_convergence}
\end{equation}
uniformly in \(s\in[0,1]\). Hence, by continuity of
\((x,\eta)\mapsto
\nabla_{x}\phi (x,\eta)\), the integrand in
\eqref{eq:A_ratio} converges to \(0\) uniformly in \(s\in[0,1]\).

Hence
\(
\frac{|A(\mathbf u)|}{\|\mathbf u\|}\to0,
\)
and so
\begin{equation}
A(\mathbf u)=o(\|\mathbf u\|).
\label{eq:A_little_o}
\end{equation}

The proof for $B(\mathbf u)$ is similar. Indeed, applying the same argument to
\eqref{eq:B_integral}, and using the continuity of
\((x,\eta)\mapsto
\nabla_{x}\phi (x,\eta)\), we obtain
\(
\frac{|B(\mathbf u)|}{\|\mathbf u\|}\to0.
\label{eq:B_ratio_to_zero}
\)
Therefore,
\begin{equation}
B(\mathbf u)=o(\|\mathbf u\|).
\label{eq:B_little_o}
\end{equation}

Since $A(\mathbf u)=o(\|\mathbf u\|)$~\eqref{eq:A_little_o} and
$B(\mathbf u)=o(\|\mathbf u\|)$~\eqref{eq:B_little_o}, and since
\eqref{eq:sandwich_AB} holds, it follows that
\begin{align}
g (x+\mathbf u)-g (x)
&=
\left\langle
\nabla_{x}\phi 
\bigl(x,v(x)\bigr),
\mathbf u
\right\rangle
+
o(\|\mathbf u\|)
\label{eq:final_expansion}
\end{align}
as \(\mathbf u\to 0\). Hence \(g \) is differentiable at
\(x\), with
\begin{equation}
\nabla_{x}g (x)
=
\nabla_{x}\phi \bigl(x,v(x)\bigr).
\label{eq:gradient_g}
\end{equation}

\paragraph{\textbf{5. Continuity of the gradient}.}
\label{sec:cont_gradient}
We now prove that \(x\mapsto \nabla_x g(x)\) is continuous where
\[
\nabla_x g(x)=\nabla_x\phi(x,v(x)),
\]
Having shown that \(g \) is differentiable on \(O\) in Step~\ref{sec:diff_cvar}, it remains to prove
that its gradient is continuous. Let \(x_n\to x\) in
\(O\). By the continuity of \(v\) on \(O\) (Step~\ref{sec:cont_var_in_proof}),
\[
v(x_n)\to v(x).
\]
Hence
\(
\bigl(x_n,v(x_n)\bigr)
\to
\bigl(x,v(x)\bigr)
\quad\text{in }O\times\mathbb R .
\)
Since Step 2 shows that
\(
(x,\eta)\mapsto
\nabla_{x}\phi (x,\eta)
\)
is continuous on \(O\times\mathbb R\), we obtain
\[
\nabla_{x}\phi 
\bigl(x_n,v(x_n)\bigr)
\to
\nabla_{x}\phi 
\bigl(x,v(x)\bigr).
\]
\[
\Rightarrow^\eqref{eq:def_g_cvar}
\nabla_{x}g (x_n)
\to
\nabla_{x}g (x).
\]
\[
\Rightarrow^\eqref{eq:def_g_cvar}
x\mapsto
\mathrm{CVaR}_{\alpha}\!\bigl(L (x)\bigr)
\text{ is cont. differentiable in }
\mathcal O.
\]
\end{proof}
\printcredits





\end{document}